\documentclass[11pt]{article}

\usepackage[letterpaper,margin=1.00in]{geometry}
\usepackage{amsmath, amssymb, amsthm, thmtools, amsfonts, dsfont}
\usepackage{bbm}

\usepackage{tikz}
\usetikzlibrary{positioning,decorations.pathreplacing}

\usepackage{cite}
\usepackage{appendix}
\usepackage{graphicx}
\usepackage{color}
\usepackage{algorithm, multicol}
\usepackage[noend]{algpseudocode}
\usepackage{epstopdf}
\usepackage[textsize=tiny]{todonotes}

\usepackage{framed}
\usepackage[framemethod=tikz]{mdframed}
\usepackage[bottom]{footmisc}
\usepackage[shortlabels]{enumitem}
\setitemize{noitemsep,topsep=3pt,parsep=3pt,partopsep=3pt}
\usepackage[font=small]{caption}
\usepackage{xspace}

\usepackage{mathtools}

\newtheorem{theorem}{Theorem}[section]
\newtheorem{lemma}[theorem]{Lemma}
\newtheorem{meta-theorem}[theorem]{Meta-Theorem}

\newtheorem{corollary}[theorem]{Corollary}
\newtheorem{proposition}[theorem]{Proposition}
\newtheorem{observation}[theorem]{Observation}
\newtheorem{definition}[theorem]{Definition}

\DeclareMathOperator{\polylog}{polylog}

\newcommand{\FullOrShort}{full}

\ifthenelse{\equal{\FullOrShort}{full}}{
    
  \newcommand{\fullOnly}[1]{#1}
  \newcommand{\shortOnly}[1]{}
  }{
    \newcommand{\fullOnly}[1]{}
    \newcommand{\shortOnly}[1]{#1}
  }

\definecolor{darkgreen}{rgb}{0,0.5,0}
\usepackage{hyperref}
\hypersetup{
    unicode=false,          % non-Latin characters in Acrobat’s bookmarks
    colorlinks=true,        % false: boxed links; true: colored links
    linkcolor=red,          % color of internal links (change box color with linkbordercolor)
    citecolor=darkgreen,        % color of links to bibliography
    filecolor=magenta,      % color of file links
    urlcolor=cyan           % color of external links
}
\usepackage[capitalize, nameinlink]{cleveref}
\crefname{theorem}{Theorem}{Theorems}
\Crefname{lemma}{Lemma}{Lemmas}
\Crefname{observation}{Observation}{Observations}
\Crefname{claim}{Claim}{Claims}
\Crefname{equation}{}{}

\algnewcommand\algorithmicswitch{\textbf{switch}}
\algnewcommand\algorithmiccase{\textbf{case}}
\algblock{ParFor}{EndParFor}
\algnewcommand\algorithmicparfor{\textbf{parfor}}
\algnewcommand\algorithmicpardo{\textbf{do}}
\algnewcommand\algorithmicendparfor{\textbf{end\ parfor}}
\algrenewtext{ParFor}[1]{\algorithmicparfor\ #1\ \algorithmicpardo}
\algtext*{EndParFor}%{\algorithmicendparfor}

\algdef{SE}[SWITCH]{Switch}{EndSwitch}[1]{\algorithmicswitch\ #1\ \algorithmicdo}{\algorithmicend\ \algorithmicswitch}%
\algdef{SE}[CASE]{Case}{EndCase}[1]{\algorithmiccase\ #1}{\algorithmicend\ \algorithmiccase}%
\algtext*{EndSwitch}%
\algtext*{EndCase}%
\usepackage{thmtools} 
\usepackage{thm-restate}

\usepackage[textsize=tiny]{todonotes}

\usepackage[capitalize, nameinlink]{cleveref}
\crefname{theorem}{Theorem}{Theorems}
\Crefname{lemma}{Lemma}{Lemmas}
\Crefname{observation}{Observation}{Observations}
\Crefname{equation}{}{}

\newcommand{\eps}{\varepsilon}

\newcommand{\poly}{{\rm poly}}
\newcommand{\sgn}{\operatorname{\text{{\rm sgn}}}}

\newcommand{\ceil}[1]{\left\lceil #1 \right\rceil}

\providecommand{\E}{{\rm \mathbb{E}}}
\providecommand{\Var}{{\rm \mathbb{V}ar}}

\newcommand{\drivestreams}{drivestreams}

\renewcommand{\paragraph}[1]{\vspace{0.15cm}\noindent {\bf #1}:}

\usepackage{setspace}

\usepackage{mathtools}
\DeclarePairedDelimiter{\abs}{\lvert}{\rvert}

\makeatletter
\let\oldabs\abs
\def\abs{\@ifstar{\oldabs}{\oldabs*}}
\makeatother

\date{}

\title{\huge Improved parallel derandomization via finite automata with applications}

\author{Jeff Giliberti\footnote{University of Maryland, email: \texttt{jeffgili@umd.edu}. JG gratefully acknowledges financial support by the Fulbright U.S. Graduate Student Program, sponsored by the U.S. Department of State and the Italian-American Fulbright Commission. This content does not necessarily represent the views of the Program.}, \ \  David G. Harris\footnote{University of Maryland, email: \texttt{davidgharris29@gmail.com}}}

\date{}

\begin{document}
\maketitle
\begin{abstract}
A central approach to algorithmic derandomization is to construct probability distributions with small support that ``fool” randomized algorithms, often enabling efficient parallel (NC) implementations. An abstraction of this idea is fooling polynomial-space statistical tests computed via finite automata~(Sivakumar 2002); this encompasses a wide range of properties including $k$-wise independence and sums of random variables.

We present new parallel algorithms to fool automata, with significantly reduced processor complexity. Briefly, our approach is to iteratively sparsify distributions via work-efficient lattice discrepancy rounding, while tracking an aggregate weighted error that is determined by the Lipschitz value of the statistical tests.

We illustrate with applications to the Gale-Berlekamp Switching Game and approximate MAX-CUT via SDP rounding. These involve several optimizations, including truncating the state space of the automata and using FFT-based convolutions to compute transition probabilities efficiently.
\end{abstract}

This is an extended version of a paper appearing in the European Symposium on Algorithm (ESA) 2025.

\section{Introduction}
A fundamental problem in the theory of computation is to \emph{derandomize} existing randomized algorithms. Such randomized algorithms typically use a large number of independent random bits. One main  approach to derandomization is to construct a probability distribution which is much smaller (of polynomial instead of exponential size), to ``fool" the randomized algorithm. That is, the behavior of relevant statistics should be similar when presented with fully-independent random bits vs. bits drawn from a small, carefully-constructed correlated distribution. This probability space can be searched exhaustively, and in parallel, to find a specific input with desired properties. 

A simple and popular example of this technique comes from a probability space with $k$-wise independence, for constant $k$ \cite{ABI86, J74, KY94, L85}. This preserves  basic statistical properties such as means and variances \cite{BR94, SSS95}.  There are many other constructions for more-advanced properties, e.g. near-$k$-wise independence \cite{AGHP92, AMN98, CRS94, GGLN92, NN90}, fooling halfspaces or polytopes \cite{KM21, OSTK20, PMRZ11, GMRTV12}, fooling ordered branching programs \cite{BRRY, CHLTW23, MRT19, HPV21, N92} and so on. Often, the randomized algorithm itself must be significantly modified or ``slowed down" to work with these spaces \cite{H19hmis}.

A significant abstraction of these methods is the construction of probability spaces that fool statistical properties (also called ``tests") computed by finite automata \cite{BG92, H19, KK97, MRS01, N90, N92, S02}. Such tests are ubiquitous in randomized algorithm analysis, encompassing $k$-wise independence, sums of random variables, and many other properties. For example, they are used in deterministic parallel algorithms for finding balanced independent sets in graphs \cite{harris2016efficient}, for applications of the Lov\'{a}sz Local Lemma in combinatorics \cite{H23}, for covering and packing integer programs \cite{BKNS12, S01}, for undirected  connectivity problems \cite{N92}, and more.

\subsection{Our contribution}
We present new, simplified parallel algorithms to fool automata, with significantly reduced processor complexity. It will require several preliminary definitions to describe how these automata work and the error bounds to use. However, we provide a summary of our algorithm performance as follows:

% \begin{theorem}[Simplified]
% \label{simpthm}
% Consider a probability space on $n$ binary variables $X_1, \dots, X_n$, and weighted automata $F_1, \dots, F_z$, each with state space of size at most $\eta$ and  Lipschitz value $\lambda$, that read in the input variables in the same order.

% There is an algorithm that produces a distribution $D$ of size $\eps^{-2} \cdot \polylog(n, \eta, z, 1/\eps)$ to simultaneously fool the automata to \emph{absolute weighted error} $\lambda \eps$. It is deterministic and uses $\tilde O(n z \eta  \eps^{-2} + n z \eta^{\omega})$ processors and polylogarithmic time, where $\omega$ is the exponent of matrix multiplication.
% \end{theorem}

\begin{theorem}[Simplified]
\label{simpthm}
Consider a probability space $\Omega$ on $n$ binary random variables and a collection of $\ell$ statistical tests over $\Omega$, each with $\eta$ possible states and Lipschitz value $\lambda$.

There is a deterministic algorithm to generate a distribution $D$ of size $\eps^{-2} \polylog(n, \eta, \ell, 1/\eps)$ to simultaneously ``fool'' all statistical tests to \emph{absolute weighted error} $\lambda \eps$. It uses $\tilde O(n \ell \eta  \eps^{-2} + n \ell \eta^{\omega})$ processors and polylogarithmic time, where $\omega$ is the exponent of matrix multiplication.
\end{theorem}

By way of comparison, the algorithm of \cite{H19} would require roughly $O(n^3 \ell^2 \eta^4 \eps^{-2})$ processors for a distribution of size $O(\ell \eta^2 \epsilon^{-2})$. Our full results are more general, easily allowing for non-binary alphabets, or more complex automata types; see \Cref{th:main1} for details. We emphasize that the algorithm is still typically more efficient even for applications with no particular Lipschitz bounds. For example, if we only know that the weight function takes values in the range $[0,1]$, then we trivially have $\lambda = 1$ and \Cref{simpthm} would be significantly stronger than prior algorithms.

We illustrate with two prototypical applications: the Gale-Berlekamp Switching Game and SDP rounding for approximate MAX-CUT. 

\begin{theorem}[Gale-Berlekamp Switching Game]
\label{th1}
Given an $n \times n$ matrix $A$, there is a deterministic parallel algorithm using $\tilde O (n^{3.5})$ processors and polylogarithmic time to find $x, y \in \{-1, +1\}^n$ satisfying $\sum_{i,j} A_{i,j} x_i y_j \geq (\sqrt{2/\pi} - o(1)) n^{3/2}.$
\end{theorem}

It is interesting that the Gale-Berlekamp analysis revolves around \emph{anti-concentration} bounds, which are precisely the opposite of discrepancy-minimization. Our algorithm  for this problem beats the complexity of the optimized algorithm of \cite{H19} which required $n^{5+o(1)}$ processors. 

\begin{theorem}[MAX-CUT Approximation]
\label{th2} 
      Given an $n$-vertex $m$-edge graph $G$, there is a deterministic parallel algorithm using $\tilde O (m n^3)$ processors and polylogarithmic time to find an $\alpha (1 - \eps)$ approximate MAX-CUT of $G$, where $\alpha \approx 0.878$ is the Goemans-Williamson approximation constant \cite{GW95} and $\eps > 0$ is an arbitrary constant.
\end{theorem}

For comparison, a \textit{sequential} deterministic $\alpha$-approximation with $\tilde O(n^3)$ runtime was shown in \cite{BK05}. 
This relies on the method of conditional expectation, which is hard to parallelize.  The work of \cite{S02} used the automata derandomization framework to get a parallel deterministic $\alpha(1 - \eps)$-approximation algorithm.  The main downside of this approach is the huge polynomial processor complexity (on the order of $n^{100}$).

The complexity bounds in \Cref{th1,th2} do not depend on fast matrix multiplication; the runtimes are valid even using naive multiplication (with $\omega = 3$).
Although the new algorithms are still not fully work-efficient, this makes progress toward practical deterministic parallel algorithms.

In a broader perspective, our construction differs from classical derandomization via (sequential) PRGs. It is known that, by connections to Boolean circuits \cite{Bor77}, a PRG for log space-bounded computation can be simulated in PRAM in polylog time with a polynomial number of processors. However, relying on the equivalence between PRAM and Boolean circuits leads to very high processors complexities (e.g., Nisan's log space-bounded PRG uses at least $\Omega(n^{45})$ processors~\cite{N90}). Therefore,  known PRGs constructions would either require high polynomials (if logspace bounded) or super-polylogarithmic runtime.  While our automata framework does not provide a PRG, it seems better suited to efficient parallel derandomization. In addition, our construction returns a distribution with seed length $O( \log 1/\eps + \log \log n \eta \ell)$, which is in line with known PRGs.

\subsection{Technical approach}

Our work follows the same general outline as previous algorithms for fooling automata, built from two main subroutines \textsc{Reduce} and \textsc{Fool}. 

This \textsc{Reduce} subroutine takes a distribution $E = D^1 \times D^2$ over automata walks and outputs a refined distribution $D$ with significantly smaller support while maintaining a weighted measure of closeness to $E$.  This leverages the connection between automata and lattice approximation noted by \cite{MRS01}, and applies the work-efficient lattice approximation algorithm of \cite{GG23}. Our key approach to reduce the processor complexity is to sparsify $E$ --- without materializing it --- in $\log_2 |E|$ steps. 

The \textsc{Fool} algorithm (\cref{sec-fool})  recursively invokes \textsc{Reduce} (\cref{sec-reduce}) to construct a fooling distribution, beginning with single-step automaton transitions and progressively scaling up to multi-step distributions. After $\log_2 n$ levels, the final distribution fools $n$-step walks in the automata.

 This resembles the PRG of \cite{INW}, which replaces the concatenation of two independent random walks with correlated random walks.
 
\paragraph{Error tracking} Prior algorithms \cite{H19, MRS01} considered a notion of unweighted absolute error, where the goal was to fool all pairs of $(\text{start, } \text{end})$ states. We replace this with an \emph{aggregated} notion: the weight corresponds to the value of the associated final outcome (e.g., the final count for a counter automaton). The value of an intermediate state is then the expected weight of the final state under a random suffix. This is similar to a scheme for pseudorandom generators and branching programs \cite{BRRY,CHLTW23}. The processor complexity is determined by the \textit{Lipschitz value} of each state: how much the expected final weight changes with a single step of the automaton. 

\paragraph{Further optimizations} In both our applications, the relevant automata are counters to accumulate the running sum of certain statistics. Instead of tracking these exactly, we reduce the state space by truncating the sums to within a few standard deviations of their means. Doing this while also preserving relevant Lipschitz properties is technically challenging and requires significant analysis. This optimization will be discussed in \cref{sec:counter}.

The approximate MAX-CUT application (\cref{sec:max-cut}) requires a few additional optimizations. In particular, to simulate the SDP rounding procedure efficiently, we (i) use discretized truncated Gaussians and quantized counters, (ii) use a pessimistic estimator with a better Lipschitz constant, and (iii) compute the transition matrices via FFT's to exploit their convolution structure.  
We believe that these optimizations may generalize to other rounding procedures and other settings.

\section{Preliminaries}\label{sec:prel}
For an input of size $N$, our goal is to develop deterministic PRAM algorithms with $\poly(N)$ processor complexity and $\polylog(N)$ runtime.  \textbf{Unless stated otherwise, we assume that all relevant algorithms run in \emph{deterministic polylogarithmic time} as a function of their input complexities}. Since we do not track runtime precisely, we overlook some of the model details (e.g. whether it uses shared or exclusive write access).

For processor complexity and other parameters, we use $\tilde O$ notation: we say that $f(x) \leq \tilde O(g(x))$ if $f(x) \leq g(x) \cdot \polylog(N,g(x))$ processors, where $N$ is the size of all algorithm inputs. 

Throughout, we use $\log$ for logarithm in base $e$ and $\lg$ for base $2$. We always assume $\eps \in (0,\tfrac{1}{2})$. 

We write $\lceil x \rfloor$ for rounding to the nearest integer.

\subsection{Basic definitions for automata}
The underlying probability space $\Omega$ is defined by $n$ independent random variables $R_0, \dots, R_{n-1}$, where each $R_t$ has a distribution $\Omega_t$ over an arbitrary alphabet (which we also denote by $\Omega_t$). We consider an automaton $F$ with state spaces $S_t: t = 0, \dots, n$. At each timestep $t = 0, 1, \dots, n-1$, the automaton in state $s \in S_t$ receives an input $R_t$ and transitions to state $F(s,R_t) \in S_{t+1}$.  We call the states in $S_0$ the \emph{starting states}. 

For times $t, t'$ with $t' > t$,  we define $\Omega_{t,t'} = \Omega_t \times \Omega_{t+1} \times \dots \times \Omega_{t' - 1}$. In this context, we call $h = t' - t$ the \emph{horizon}. We refer to a vector $\vec r = (r_t, r_{t+1}, \dots, r_{t'-1} ) \in \Omega_{t,t'}$ as a \emph{drivestream}. We define $F(s,\vec r)$ to be the result of transiting from time $t$ to $t'$ under a given drivestream $\vec r$. 

For a distribution $D$ on drivestreams, we denote the \emph{transition matrix} as $T_D^F$, i.e. $T^F_D(s,s')$ is the probability of transiting from state $s$ to $s'$ for a drivestream $\vec R \sim D$. For brevity, we also write $T^F_{t,t'} := T^F_{\Omega_{t,t'}}$ and for any weight function $w: S \rightarrow \mathbb R$ we define
$$
T_D^F(s,w) = \sum_{s' \in S_t} T_D^F(s,s') w(s').
$$

Throughout, we define  $\eta = \max_t |S_t|$ and $\sigma = \max_t |\Omega_t|$, where $|\Omega_t|$ is the size of the alphabet of $\Omega_t$. We call $\eta$ the \emph{statesize} of $F$ and $\sigma$ the \emph{alphabetsize}. 

\paragraph{Measurement of distribution error} Our goal is to find a polynomial-size distribution $D$ on drivestreams that ``fools'' the automaton. That is, the behavior of the automaton given a drivestream $\vec R \sim D$ should be similar to its behavior when given a drivestream $\vec R \sim \Omega$. 

We follow a strategy of \cite{BRRY} and measure error via a weight function $W: S_n \rightarrow \mathbb R$ over final states, keeping track of ``smoothness" of $W$ as a function of the drivestream values. 

\begin{definition}[Confusion and Total Variability]
The \emph{confusion} at state $s \in S_t$ is defined by 
$$
{\mathfrak C}^F(s) = \max_{r_t, r_t' \in \Omega_t} \E_{R_{t+1}, \dots, R_n \sim \Omega_{t,n}} \Bigl[ \bigl| W( F(s, (r_t, R_{t+1}, \dots, R_n)) - W( F( s,(r_t', R_{t+1}, \dots, R_n))) \bigr| \Bigr]
$$
i.e., the maximum \emph{expected change} in the  final state weight due to changing drivestream entry $t$.

The \emph{total variability} ${\mathfrak V}^F(s)$ of a starting state $s$ is defined by
\begin{align*}
 {\mathfrak V}^F(s) &:= \sum_{h=1}^{n} \max_{\vec r \in \Omega_{0,h}} \mathfrak C^F(   F( s, \vec r) )
 \end{align*}
\end{definition}

Note that $\mathfrak C^F(s)$ is upper-bounded by the \emph{Lipschitz value} of $W(F(s,\vec r))$ as a function of $r_t$, i.e. the maximum change in the weight of the final state of due to changing coordinate $r_t$.

With this rather technical definition, we can state our algorithm concretely:
\begin{theorem}
\label{th:main1}
Our parallel automata-fooling algorithm (\cref{alg-fool}) takes input parameter $\eps > 0$, and produces a distribution $D$ of size $ \eps^{-2} \polylog(n, \eta, \sigma, 1/\eps) $ with
$$
|T_D^F(s,W) - T_{\Omega}^F(s,W)| \leq \eps \thinspace {\mathfrak V}^F( s )  \qquad \text{for all states $s \in S_0$.}
$$

It uses $\tilde O \bigl(  n \eta / \eps^2 + n \eta \sigma \bigr)$ processors, plus the cost of computing the expectations $T^F_{t,n}(s,W)$ for states $s \in S_t$ (more details on this step to be provided later).
\end{theorem}

When $F$ is clear, we write just $T^D, \mathfrak C, \mathfrak V$ etc.

\subsection{Comparison with other automaton models}\label{sec:comp_aut}
Our automaton model is slightly more general compared to previous algorithms. We highlight a few key differences and how to handle them in our framework.

\paragraph{$\bullet$ Nonuniform probability spaces} Previous constructions (e.g. \cite{MRS01, S02}) required the underlying probability space $\Omega$ to be  the uniform distribution on $\{0,1 \}^n$. This can be used to simulate other input distributions, at the cost of problem-specific error analysis or other computational overheads. Our framework, with arbitrary distributions $\Omega_t$, can simplify these constructions. For example, in our application to MAX-CUT (\Cref{sec:max-cut}), we can directly deal with Gaussian random variables vs. simulating them through Binomial distributions as in \cite{S02}.

\paragraph{$\bullet$ Multiple automata} Typical applications have multiple statistical tests to run on the data, and we want a distribution fooling them all simultaneously. Concretely, suppose that for each $i = 1, \dots, \ell$ there is a separate automaton $F_i$ on states $S_{i,0}, \dots, S_{i,n}$ of statesize $\eta_i$;  crucially, all these automata  read the data in the same order $R_0, \dots, R_{n-1}$.  We handle this by encoding them all into a single large automaton, with state space $S_t = \{ (i,s): s \in S_{i,t} \}$ and statesize $\eta = \sum_i \eta_i$. Note that states $(i,s)$ will never mix with states $(i',s')$ for $i' \neq i$.  

\paragraph{$\bullet$ Leveled automata} Previous construction sometimes use additional clock inputs; that is, the automaton receives input $(R_t,t)$ for each time $t$. This is sometimes called a ``leveled automaton". Our use of distinct state spaces $S_t$ for each time $t$  makes this unnecessary: when receiving the input $s \in S_t$, the automaton will automatically learn the time $t$.

\section{Overview of algorithms}\label{sec:aut}
The subroutine \textsc{Reduce} (\cref{sec-reduce}) is the technical core of the automata-fooling construction: given an input distribution $E$ over drivestreams and a weight function $w$, it returns a distribution $D$ which is close to $E$ (measured in terms of $w$), but which has much smaller support. (Note that $w$ will not be the given weight function over final states $W$.) Importantly, through the use of appropriate data structures, the cost of \textsc{Reduce} may be significantly less than the total size of $E$ itself. Later, the final algorithm \textsc{Fool} (\cref{sec-fool}) will be defined by repeated calls to \textsc{Reduce} over successively larger time-horizons. 

 \subsection{Data structures for distributions}
 
Concretely, we store each distribution $D$ as an array of drivestreams $D[0], \ldots, D[\ell-1] \in \Omega_{t,t+h}$ with associated probabilities $p_D(0), \dots, p_D(\ell-1)$. Here $\ell = |D|$ is the \emph{size} of $D$. By adding dummy zero entries as needed, we always assume that $|D|$ is a power of two.

For a bitstring $b$ of length at most $\lg|D|$,  we denote by $D[b]$ the induced distribution consisting of all the \drivestreams\ in $D$ whose indices start with $b$. So $D[b] = D[b \# 0] \cup D[b \# 1]$, where $\#$ denotes bit concatenation  Similarly, we define $p_{D}(b) = \sum_{a \in D[b]} p_D(a)$.

\begin{observation}
\label{step-obs}
Given a drivestream $\vec r \in \Omega_{t,t+h}$, we can compute all values $F(s,\vec r): s \in S_t$ using a total of $\tilde O(\eta h)$ processors.
\end{observation}
\begin{proof}
We do it recursively: we compute the values for times $t$ and $t + h/2$ with drivestreams $\vec r^1 = (r_t, \dots, r_{t+h/2-1}), \vec r^2 = (r_{t+h/2}, \dots, r_{t+h-1})$. Then, we compute $F(s,\vec r) = F( F(s,\vec r^1), \vec r^2 )$ in parallel for all $s$. The result follows by recalling $|S_i| \leq \eta$ for all $i$.
\end{proof}

We use the following key task for our algorithm.
\begin{definition}[Prediction Problem]
Given a distribution $D \subseteq \Omega_{t,t+h}$ and a weight function $w: S_{t+h} \rightarrow \mathbb R$,  we need to produce a data structure $\mathcal Q(D,w)$ to answer the following two types of queries: (i) given any bitstring $b$, return the value $p_D(b)$; (ii) given $(b,s)$ where $b$ is a bitstring and $s$ is a state in $S_t$, return the value $T_{D[b]}(s,w)$. 

Each query should take polylogarithmic time and  processors. In either case, $b$ can take on any length $\ell \leq \lg|D|$. 
\end{definition}

\begin{observation}
\label{pred-prob-gen}
The Prediction Problem for any distribution $D$ of horizon $h$ and any weight function $w$ can be solved with $\tilde O( |D|  h \eta)$ processors.
\end{observation}
\begin{proof}
To support query of type (i), for each bitstring $b$, we can look up all the corresponding distribution entries $p_D(a)$ with $a$ a suffix of $b$. Since each distribution entry appears in $\lg|D|$ bitstrings, this computation takes $\tilde O(|D|)$ processors. To support queries of type (ii), for each entry $\vec r = D[a]$ and every state $s \in S_t$, we use \Cref{step-obs} to determine $s' = F(s,\vec r)$ and add the contribution $p_{D[b]}(a)w(s')$ for every prefix $b$ of $a$.
\end{proof}

The most important case is for a Cartesian product. This result will be used later in \cref{sec-fool}.

\begin{proposition}\label{pred-prob}
Given distributions $D^1 \subseteq \Omega_{t,t+h}$ and $D^2 \subseteq \Omega_{t+h,t+2 h}$, the Prediction Problem for the Cartesian product distribution $E = D^1 \times D^2$ and any weight function $w$  can be solved with $ \tilde O( \eta h ( |D^1| + |D^2| )) $ processors (without ever explicitly materializing $E$).
\end{proposition}
\begin{proof}
First, use \Cref{step-obs} to compute the values $F(s,\vec r)$ for all states $s$ and drivestreams $\vec r \in D^1$. Then, use \Cref{pred-prob-gen} to solve the Prediction Problem for $D^2, w$. Using $\mathcal Q(D^2,w)$,  compute the weight function
$
w_2(s) = T_{D^2}(s,w).
$
and finally apply \Cref{pred-prob-gen} a second time to solve the Prediction Problem for $D^1, w_2$.

Let $L_1 = \lg|D^1|$. Now consider a bitstring $b$ of some length $\ell$. If $\ell \leq L_1$, then
$$
p_{E}(b) = p_{D^1}(b), \qquad T_{E[b]}(s,w) = T_{D^1[b]}(s,w_2)
$$
for any state $s$; such queries can be answered using the data structure $\mathcal Q(D^1, w_2)$. Otherwise, if $\ell > L_1$, then let $b = (b_1, b_2)$ where $b_1$ has length $L_1$ and $b_2$ has length $\ell - L_1$. We have
$$
p_{E}(b) = p_{D^1}(b_1) \cdot p_{D^2}(b_2),\qquad T_{E[b]}(s, w) = T_{D^2[b_2 ]}( F(s,D^1[b_1]) , w) 
$$
which can be computed by using the data structure $\mathcal Q(D^2, w)$ and the values $F(s,\vec r): \vec r \in D^1$.
\end{proof}

\subsection{Lattice approximation}
The problem of automata fooling is closely linked to lattice approximation, defined as follows:
\begin{definition}[Lattice Approximation Problem (LAP)]
    \label{def-LAP}
    Given an $m \times n$ matrix $A$ and a fractional vector $\vec u \in [0,1]^n$, the objective is to compute an integral vector $\vec v \in \{0,1 \}^n$ to minimize the discrepancies $D_k = |\sum_{j=1}^{n} A_{kj}(u_j - v_j)|$.
\end{definition}
Intuitively, this models the process of ``rounding" each random bit (represented by $u_j$) to a static zero or one (represented by $v_j$). The work \cite{GG23} provides a parallel algorithm with near-optimal discrepancy as well as complexity; we summarize it as follows:
\begin{theorem}[Theorem 1.3 of \cite{GG23}]
\label{thm-lattice}
Suppose that $A_{kj} \in [0,1]$ for all $k,j$. 
    There is a deterministic parallel algorithm for the LAP with $D_k \leq O(\sqrt{\mu_k \log m} + \log m)$ for all $k$, where $\mu_k = \sum_{j=1}^n A_{kj}u_j$. The algorithm uses $\tilde O(n + m + nnz(A))$ processors, where $nnz(A)$ denotes the number of non-zero entries in the matrix $A$.
\end{theorem}

For our purposes, it will be convenient to use a slight extension of \Cref{thm-lattice} which allows for the discrepancy matrix $A$ to take arbitrary values.
\begin{proposition}
\label{thm-lattice2}
Let $\Delta_k = \max_j |A_{kj}|$ for each row $k$. 
    There is a deterministic parallel algorithm  for the LAP where $D_k \leq O( \sqrt{\Delta_k \tilde \mu_k \log m} + \Delta_k \log m) \leq O( \Delta_k (\sqrt{n \log m} + \log m))$ 
    for all $k$, where $\tilde \mu_k = \sum_{j=1}^n |A_{kj}| u_j$. The algorithm uses $\tilde O(n + m + nnz(A))$ processors.
\end{proposition}
\begin{proof}
Construct a $2m \times n$ matrix $\tilde A$, where for each row $k = 0, \dots, m-1$ of  $A$, the matrix $\tilde A$ has $\tilde A_{2k, j} = \max\{0, \frac{A_{kj}}{\Delta_k} \}$ and 
$\tilde A_{2k+1,j} = \max\{0, \frac{-A_{kj}}{\Delta_k}\}$. Then apply \Cref{thm-lattice} to $\tilde A$.
\end{proof}

\section{The REDUCE Algorithm}\label{sec-reduce} 
To reiterate, the goal of \textsc{Reduce} (\cref{alg-reduce1}) is to take an input distribution $E \subseteq \Omega_{t,t+h}$ and produce a smaller approximating distribution $D$.  There is a simple randomized procedure for this task: draw $m$ elements $\vec r^1, \dots, \vec r^m$ independently with replacement from $E$, and set $D$ to be the uniform distribution over these drivestreams.  The  \textsc{Reduce} algorithm derandomizes this via a slowed-down simulation, iteratively computing distributions $D_0 = E, D_1, \ldots, D_{\ell} = D$ by fixing the $i^{\text{th}}$ bit level of each entry in $D_{i-1}$.

\begin{algorithm}[H]
    \caption{\textsc{Reduce}($E,  \eps,  w$)}
    \label{alg-reduce1}
    \centering
    \begin{algorithmic}[1]
    \State Set $\ell = \lg|E|$ and $m = \frac{C  \ell^2 \log \eta}{\eps^2}$ for a constant $C$.
    \State Solve Prediction Problem for distribution $E$    

    \medskip

    \State Initialize the multiset $H_0 := \{ m$ empty bitstrings\} 
    \For{$i = 0, \ldots, \ell - 1$}\Comment{Fix $i^{\text{th}}$ bit}

        \State Formulate LAP $\mathcal{L}$ for $H_{i}$.
        \State $\vec v \leftarrow $ solve $\mathcal{L}$ via \cref{thm-lattice2}  with sampling rate $u_b = \frac{p_E(b \# 1)}{p_E(b)}$ for each $b \in H_{i}$
        \State $H_{i+1} \leftarrow \{ b \# v_b : b \in H_{i} \}$ \Comment{ Concatenate vector $\vec v$ as $i^{\text{th}}$ bit level}
   
    \EndFor
    \medskip
    \ParFor{$j \in \{0, \dots, m-1 \} $}
    \Comment{Convert bitstring (index) to drivestream (value)}
    \State Set $D[j] = E[ H_\ell[j] ]$ and set probability $p_D[j] = 1/m$.
    \EndParFor
    \Return $D$ \textcolor{white}{$\left\{D^1\right\}$}
    \end{algorithmic}
\end{algorithm}

To measure distribution error, we introduce the following key definition:

\begin{definition}[Spread]
For $w: S_{t+h} \rightarrow \mathbb R$,  the \emph{spread $\mathfrak S^F_E(s,w)$} of state $s \in S_t$ is:
$$
\mathfrak S^F_E(s,w) :=  \max_{\vec r^1 \in E} w(F(s,\vec r^1)) - \min_{\vec r^2 \in E} w(F(s,\vec r^2))
$$

We write $\mathfrak S^F_{t,t+h}$ as shorthand for $\mathfrak S^F_{\Omega_{t,t+h}}$.
\end{definition}

The following is our main result analyzing the algorithm:
\begin{theorem}
\label{thm-reduce-additive}
Algorithm \textsc{Reduce} runs in $\tilde O(\frac{h + \eta}{\eps^2})$ processors, plus the cost of solving the Prediction Problem for $E$. The final distribution $D$ is uniform with size $O( \frac{ \log \eta \log^2 |E| }{\eps^2})$, and it satisfies
$$
| T_D^F(s,w) - T_E^F(s,w) | \leq \eps \thinspace \mathfrak S^F_E(s,w) \qquad \text{for each state $s \in S_t$}
$$
\end{theorem}
\begin{proof}
Let us fix $F, E,w$; for brevity,  we write $\alpha_{s} := \mathfrak S^F_E(s,w)$ for each state $s$.

Each multiset $H_i$ contains $m$ bitstrings $H_{i}[0], \dots, H_{i}[m-1]$ of length $i$. Consider 
the distribution $D_i$ obtained by drawing a bitstring $b \in H_i$ uniformly at random and then drawing a drivestream from $E[b]$. In particular, $D_0 = E$ and the distribution returned by \textsc{Reduce} is precisely $D_\ell = D$. We have the following equation for every state $s$:
\begin{align*}
    T_{D_{i}}(s,w) &= \frac{1}{m} \sum_{b \in H_{i}} T_{E[b]}(s,w).
\end{align*}

Observe that $H_{i+1}$ is obtained by appending a bit $v_b$ to each bitstring $b \in H_i$. 
We expose the choice of the next bit in $D_i$ and $D_{i+1}$ as follows
\begin{align*}
    T_{D_{i}}(s,w) &= \frac{1}{m} \sum_{b \in H_{i}} \left[ \tfrac{p_E(b\#1)}{p_E(b)}  \cdot T_{E[b\#1]}(s,w) + \tfrac{p_E(b\#0)}{p_E(b)} \cdot T_{E[b\#0]}(s,w)\right]\\
        T_{D_{i+1}}(s,w) &= \frac{1}{m} \sum_{b \in H_{i}}  \left[ v_b \cdot T_{E[b\#1]}(s,w) + (1-v_b) \cdot T_{E[b\#0]}(s,w)\right].
\end{align*}
Thus we can calculate the difference between probabilities for $D_i$ and $D_{i+1}$ as:
\begin{align}
\label{frg1}
   T_{D_{i+1}}(s,w) - T_{D_i}(s,w) &= \frac{1}{m} \sum_{b \in H_{i}}  \bigl( \tfrac{p_E(b\#1)}{p_E(b)} - v_b \bigr) \bigl( T_{E[b\#1]}(s,w) - T_{E[b\#0]}(s,w) \bigr).
\end{align}

In light of Eq.~(\ref{frg1}), we apply \Cref{thm-lattice2} where each state $s$ corresponds to a constraint row $k$ with entries $A_{kb} = T_{E[b\#1]}(s,w) - T_{E[b\#0]}(s,w )$, 
and with $u_b = \tfrac{p_E(b\#1)}{p_E(b)}$ for all $b$; note that, after solving the Prediction Problem for $E$, all these values can be generated using $O(\eta m)$ processors. The maximum spread of the values $w(F(s,\vec r))$ over all possible choices of $\vec r$ is at most $\alpha_s$, so
$$
\Delta_k = \max_k |A_{kb}| \leq \alpha_s
$$
Since the matrix has $\eta$ rows and $m$ columns,   \Cref{thm-lattice2} gives:
$$
| T_{D_{i+1}}(s,w) - T_{D_i}(s,w) | \leq O \Bigl( \frac{ \alpha_s \sqrt{ m \log \eta} + \alpha_s \log \eta}{m} \Bigr)
$$

By our choice of $m$, this is at most $\eps \alpha_s/\ell$. Over all iterations, this gives the desired bound \[
 | T_D(s,w) - T_E(s,w)| \leq \sum_{i=0}^{\ell-1} |T_{D_{i+1}}(s,w) - T_{D_{i}}(s,w) | \leq \eps \alpha_s. \qedhere
\]
\end{proof}

\section{The FOOL Algorithm}\label{sec-fool}
We build the automata-fooling distribution via the algorithm $\textsc{Fool}$ (\cref{alg-fool}). 

\begin{algorithm}[H]
\caption{\textsc{Fool}($\eps$) (assume $n$ is a power of two)}
\label{alg-fool}
\centering
\begin{algorithmic}[1]
\State Set $D_{0,t} = \Omega_t $ for each $t = 0, \dots, n-1$
\State Determine approximate transition vectors $\hat V_t(s) \approx T_{t,n}^F (s,W): s \in S_t, t = 0, \dots, n$.
\medskip
\For{$i = 0, \dots, \lg n - 1$}
\ParFor{$t \in \{ 0, 2 h, 4 h, 8 h, n - 2 h \}$ where $h = 2^i$} 
  \State $D_{i+1,t} \leftarrow \textsc{Reduce}(  D_{i,t} \times D_{i,t+h}, \delta, \hat V_{ t + 2 h} )$ for $\delta = \frac{\eps}{20 (1 + \lg n)}$
\EndParFor
\EndFor
\Return final distribution $D = D_{\lg n,0}$ 
\end{algorithmic}
\end{algorithm}

The algorithm first finds the ``expected weight" vectors $V_t$ for each distribution $\Omega_{t, n}$. This may use matrix multiplication, or it may take advantage of  problem-specific structures. It may also have some small error. Then, the main loop fools distributions on time horizons $h = 1, 2, 4, 8, \dots $ in a bottom-up fashion, merging them with the \textsc{Reduce} procedure from the previous section.   Specifically, at each level $i$, it uses \textsc{Reduce} to combine the distributions $D_{i,t}$ and $D_{i,t+h}$ into the joint distribution $D_{i+1,t}$.  We never materialize the Cartesian product $D_{i,t} \times D_{i,t+h}$.

\begin{theorem}
\label{fool-cost}
The cost of \textsc{Fool} is $\tilde O \bigl(  \frac{n \eta}{\eps^2} + n \eta \sigma \bigr)$
processors, plus the cost of computing the vectors $\hat V_t: t=0,\dots, n$. The final distribution $D_{\lg n,0}$ has size $O( \frac{\log^5( n \eta \sigma / \eps)}{\eps^2})$.
\end{theorem}
\begin{proof} 
Let $N = \max\{\eta, \sigma, n, 1/\eps \}$. We first show by induction on $i$ that $|D_{i,t}| \leq \eps^{-2} \log^5 N$ for all $i \geq 1$ and large enough $N$. For, let $D^1 = D_{i,t}, D^2 = D_{i,t+h}$ and $D^{12} = D_{i+1,t}$. By induction hypothesis (and taking into account the distribution sizes for $i = 0$), distributions $D^1, D^2$ have size at most $\sigma + (\log^5 N)/\eps^2$. So $|D^1 \times D^2| \leq \sigma^2 (\log^{10} N)/\eps^4$. By specification of \textsc{Reduce}, we get $|D^{12}| \leq O(\delta^{-2} \log \eta \cdot \log^2(  \sigma^2 (\log N)^{10}/\eps^4 ))$, which is indeed at most $\eps^{-2} \log^5 N$ for large enough $N$. 

Each call to \textsc{Reduce} at level $i$ has processor complexity $\tilde O ( ( 2^i + \eta) \delta^{-2} )$ by \Cref{thm-reduce-additive}, plus the cost of solving the Prediction Problem on distribution $D_{i,t} \times D_{i,t+h}$. For $i \geq 1$, the latter step has cost $\tilde O( \delta^{-2} \eta 2^i)$.  There are $n/2^i$ calls to \textsc{Reduce} at level $i$, so the overall processor complexity $\frac{n}{2^i} \cdot \tilde O( (2^i\eta) \cdot \delta^{-2} ) = \tilde O(n \eta/\eps^2 )$. Similarly, by \Cref{pred-prob-gen}, the Prediction Problems at level $i = 0$ have total cost $\tilde O(n \eta \sigma)$.
\end{proof}

We now proceed to analyze the error of the final distribution $D$. For purposes of analysis, we define the \emph{exact} transition vector $V_t = T_{t,n}(s,W)$ and its approximation error by 
$$
\beta = \max_{s,t} |V_t(s) - \hat V_t(s)|
$$

\begin{proposition}\label{prop-add-err0s}
For any time $t$, horizon $h$ and state $s \in S_t$, we have
\begin{align*}
\mathfrak S_{t,t+h}(s,  \hat V_{t+h}) &\leq 2 \beta + 2 \sum_{k=t+1}^{t+h} 
\max_{ \vec r \in \Omega_{t,k}}  \mathfrak C(F(s,\vec r))
\end{align*}
\end{proposition}
\begin{proof}
Similar bounds were shown in \cite{BRRY}; for completeness, we show it in \Cref{app44}.
\end{proof}

\begin{theorem}\label{prop-add-err-add}
For any starting state $s$, the final distribution $D$ returned by \textsc{Fool} satisfies
$$
| T_{D}(s,W) - T_{\Omega}(s,W) | \leq \eps \thinspace {\mathfrak V}(s) + 3 \beta n 
$$
\end{theorem}
\begin{proof}
For a state $s \in S_t$ and $k \geq t$, let $a_{s,k}$ be the maximum value of $\mathfrak C(s')$ over states $s' \in S_k$ reachable from $s$. We will show by induction on $i$ that each distribution $D_{i,t}$ and state $s \in S_t$ satisfies
\begin{equation}
\label{cr3}
| T_{D_{i,t}}(s,V_{t+2^i}) - T_{t,t+2^i}(s,W) | \leq  3 ( 2^i - 1)  \beta +  2 i \delta \sum_{k=t}^{t+2^i-1} a_{s,k} 
\end{equation}

The base case $i = 0$ is vacuous since $D_{0,t} = \Omega_t$. The final case  $i =  \lg n, t = 0$ establishes the claimed result, since $\delta = \frac{\eps}{20 (1 + \lg n)}$ and ${\mathfrak V}(s) = \sum_{t=0}^{n-1} a_{s,t}$ for $s \in S_0$. For the induction step $i$, write $h = 2^i, \delta' = 2 i \delta, \beta' = 3 ( 2^{i} - 1)  \beta$, and $D^1 = D_{i,t}, D^2 = D_{i,t+h},  D^{12} = D^1 \times D^2$. For each state $s$, we have:
$$
T_{D^{12}}(s,V_{t+2h})  =  \sum_{s' \in S_{t+h}} T_{D^1}(s,s') T_{D^2} (s', V_{t+2h}).
$$

By induction hypothesis, and the fact that $a_{s',k} \le a_{s,k}$ for $s'$ reachable from $s$, this is at most
\begin{align}&
\label{cg0}
\sum_{s'} T_{D^1}(s,s') \Bigl( V_{t+h}(s') + \beta' + \delta' \sum_{k=t+h}^{t+2h-1} a_{s,k}  \Bigr) = T_{D^1}(s,V_{t+h}) + \beta' + \delta' \sum_{k=t+h}^{t+2h-1} a_{s,k} 
\end{align}

From the induction hypothesis a second time, this in turn is at most 
\begin{align}
\label{cg0a}
V_t(s) + \beta' + \delta' \sum_{k=t}^{t+h-1} a_{s,k} + \beta' + \delta' \sum_{k=t+h}^{t+2h-1} a_{s,k} \leq V_t(s) + 2 \beta' + \delta' \sum_{k=t}^{t+2h-1} a_{s,k} 
\end{align}

A completely analogous calculation gives a lower bound on $T_{D^{12}}(s,V_{t+2h})$, and overall we have 
\begin{equation}
\label{cg1}
| T_{D^{12}}(s,V_{t+2h}) - V_t(s)| \leq 6 (2^i - 1)  \beta + 2 i \delta \sum_{k=t}^{t+2h-1} a_{s,k}
\end{equation}

By \cref{thm-reduce-additive}, the distribution $D_{i+1,t}$ after applying \textsc{Reduce} has
\begin{equation}
\label{cg2}
| T_{D_{i+1,t}}(s, \hat V_{t+2h}) - T_{D^{12}}(s,\hat V_{t+2h}) | \leq \delta \thinspace \mathfrak S_{t,t+2h}(s,\hat V_{t+2h}) 
\end{equation}

Since $V_{t+2h}$ and $\hat V_{t+2h}$ differ by at most $\beta$, (\ref{cg2}) implies that
\begin{equation}
\label{cg2a}
| T_{D_{i+1,t}}(s, V_{t+2h}) - T_{D^{12}}(s, V_{t+2h}) | \leq 2 \beta + \delta \thinspace \mathfrak S_{t,t+2h}(s,\hat V_{t+2h}) 
\end{equation}

\Cref{prop-add-err0s} combined with (\ref{cg2a}) implies that
\begin{equation}
\label{cg3}
| T_{D_{i+1,t}}(s,   V_{t+2h}) - T_{D^{12}}(s,  V_{t+2h}) | \leq 2 \beta +  \delta (2 \sum_{k = t}^{t+2 h-1} a_{s,k} + 2 \beta)
\end{equation}

Combining bounds (\ref{cg1}), (\ref{cg3}) establishes the overall bound:
\begin{align*}
|T_{D_{i+1,t}}(s, V_{t+2 h}) - V_t(s)| &\leq 6 (2^i - 1) \beta + 2 i \delta \sum_{k=t}^{t+2h-1} a_{s,k} + 2 \beta + \delta (2 \beta + 2 \sum_{k = t}^{t+2 h-1} a_{s,k} ) \\
&= ( 6 (2^i - 1) + 2 + \delta) \beta + 2 (i + 1) \delta \sum_{k=t}^{t+2h-1} a_{s,k} .
\end{align*}

As $\delta < 1/2$ by our assumption on $\eps$, we have  $6 (2^i - 1) + 2 + \delta \leq 3 (2^{i+1} - 1)$. This establishes the induction bound (\ref{cr3}) for the next iteration $i+1$ as desired.
\end{proof}

As we have mentioned, there can be problem-specific shortcuts to compute (or approximate) the transition matrices $T$ and resulting expected-value vectors $V_t$. There is also a generic method to compute them via matrix multiplication; this is good enough for many applications.

\begin{proposition}\label{prop-generic_omega}
All vectors $V_t = T_{t,n}^F(s, W)$, can be computed with $\tilde O( n \sigma +  n \eta^{\omega})$ processors, where $\omega$ is the exponent of any efficiently-parallelizable matrix-multiplication algorithm. 

Moreover, if we are given the transition matrices $T_{t,t+h}^F : h = 2^i, t = j 2^i$, then we can compute all vectors $V_t$ with $\tilde O(n \eta^2)$ processors.
\end{proposition}
\begin{proof}
We begin by recursively computing transition matrices $T_{t,t+h}$, for $h = 2^i$ and $t$ divisible by $h$, over rounds $i = 0, \dots, \lg n$. At level $i = 0$, we do it by simply summing over $\Omega_{t}$. At level $i > 0$, we can compute each $T_{t,t+2h}$ via the matrix product $T_{t,t+2h} = T_{t, t+h} T_{t+h, t+2h}$, using $\tilde O( \eta^{\omega} )$ processors. The total work at the $i^{\text{th}}$ level is $\tilde O( n/2^i \cdot \eta^{\omega} )$.   

Given these transition matrices, we compute the vectors $V_t = T_{t,n} W$ in a top-down fashion: at level $i = \lg n, \dots, 0$, we compute $T_{t,n} W$ for all $t$ divisible by $h = 2^i$.  For the initial level $ i = \lg n$, we simply obtain $T_{0,n} W$ by multiplying the computed matrix $T_{0,n}$ by the vector $W$. For $i < \lg n$, we compute the vectors $T_{t,n} W = T_{t, t+h} ( T_{t+h, n} W)$ in parallel for  $t \equiv 2^{i-1} \mod 2^i$; here the matrix $T_{t,t+h}$ was computed in the first phase and the vector $T_{t+h, n} W$ was computed at iteration $i+1$ of the second phase. 
\end{proof}

\begin{corollary}
\label{multi-aprop}
Consider statistical tests $i = 1, \dots, \ell$, each computed by its own automaton $F_i$ with statesize $\eta_i$. Running \textsc{Fool} on the resulting combined automaton $F$ has processor complexity $$
\tilde O \bigl(   (n / \eps^2 + n \sigma) \sum_i \eta_i  + n \sum_i \eta_i^{\omega} \bigr)
$$
\end{corollary}

\paragraph{Remark}  The Coppersmith-Winograd algorithm \cite{coppersmith1987matrix} gives $\omega \leq 2.38$. (See \cite{jaja1992} for further details on parallel implementation.)
However, fast matrix multiplication, while theoretically attractive, is notoriously impractical. It is often preferred to use \textit{combinatorial algorithms}, which essentially means using only  ``naive" cubic matrix multiplication (i.e. with $\omega = 3$). In all the application in this paper, the matrix multiplication is not the bottleneck. We can use cubic matrix multiplication with no overall increase in the work factor.

\section{Reducing the state space for counter-like automata}
\label{sec:counter}
An automaton may have certain states which are rarely reached for random inputs. The paradigmatic example is  an automaton which tracks a sum of some statistic; there, the running sum will be concentrated in a narrow band near the mean, surrounded by a much larger set  of potential values which have exponentially small probability. 

We would like to simply discard such rare states. However, done directly, this would destroy the Lipschitz properties of the statistic: there would be a hard transition as we cross the boundary between discarded rare states and their nearby close-to-rare states.  To handle this, we need to smoothly interpolate between common and rare states.  

We define a \emph{guard function $G$} to measure the rarity of each state $s$. Formally, for a Lipschitz vector $\vec \lambda \in [0,1]^n$ and $\delta \in [0,1]$, the function $G: S_0 \cup \dots \cup S_n \rightarrow \mathbb R_{\geq 0}$ must satisfy the following conditions:

\begin{enumerate}
\item[(G1)] For all starting states $s \in S_0$ we have $G(s) = 0$.
\item[(G2)] For any state $s \in S_t$ and drivestreams $\vec r^1, \vec r^2 \in \Omega_{t,t+h}$ differing in only coordinate $t$, there holds $|G(F(s,\vec r^1)) - G(F(s,\vec r^2 )) | \leq \lambda_t \leq 1$.
\item[(G3)] For any state $s \in S_t$ and horizon $h$, we have $\Pr_{\vec R \sim \Omega_{t,t+h}} \bigl( | G(s) - G(F(s,\vec R))| > 1 \bigr) \leq \delta/n.$
\end{enumerate}

Our strategy will be to construct a truncated automaton $\tilde F$, which only stores a subset of the states. The automaton $\tilde F$ has state space $$
\tilde S_t = \{ \bot \} \cup \{ s \in S_t: G(s) < 100 \};
$$
and uses the transition rule
$$
\tilde F(s,r) = 
\begin{cases}
F(s, r) & \text{if $s \neq \bot$ and $G(F(s, r)) < 100$} \\
\bot & \text{otherwise}
\end{cases}
$$

We define a related potential function $\phi: \tilde S \rightarrow [0,1]$ as 
$$
\phi(s) = 
\begin{cases}
    0 & \text{if $s = \bot$ or $G(s) \geq 90$} \\
       1  &\text{if $G(s) \leq 10$}\\
        9/8 - G(s)/80  & \text{if $10 < G(s) < 90$} 
        \end{cases}
        $$
and we define the weight function $\tilde W$ for automaton $\tilde F$ by
$$
\tilde W(s) = \begin{cases}
\phi(s)  W(s)  & s \neq \bot \\
0 & s = \bot
\end{cases}
$$

For purposes of analysis, the potential function $\phi$, and weight function $\tilde W$, can also be applied on the original state space $S$, with the convention that $\phi(s) = 0$ for any state $s$ with $G(s) \geq 100$.

\begin{proposition}
\label{couple-prop}
For any state $s \in S_t$, let $\vec R$ be a drivestream drawn randomly from $\Omega_{t,n}$. Let $z = F(s,\vec R), \tilde z = \tilde F(s,\vec R)$ be the resulting final states for automata $F, \tilde F$ respectively.
\begin{enumerate}
\item If $G(s) \geq 91$, then $\Pr(  \phi(z) = \phi( \tilde z) = 0) \geq 1 - \delta$.
\item If $G(s) \leq 99$, then $\Pr(  z = \tilde z \neq \bot) \geq 1 -   \delta$.
\item If $G(s) \leq 9$, then $\Pr(  z = \tilde z \text{ and } \phi(z) = 1) \geq 1 -  \delta$.
\end{enumerate}
\begin{proof}
For the first result, property (G3) implies that $|G(s) - G(z)| \leq 1$ holds with probability at least $1 - \delta/n$. Suppose it does so hold; in this case, we have $G(z) \geq 90$ and hence $\phi(z) = 0$. Either $\tilde z = z$ or $\tilde z = \bot$; in either case, we have $\phi(\tilde z) = 0$.

For the next two results, let $s_i$ denote the state at each time $i$. Property (G3) and a union bound imply that, with probability at least $1 - \delta$, there holds $|G(s_i) - G(s)| \leq 1$ for all $i$. Suppose this holds. Then $s_i$ never reaches a state with $G(s_i) > 100$. So $\tilde z \neq \bot$ and hence $\tilde z = z$. Furthermore, property (G1) implies that   $G(z) = G(s_n) \leq G(s_0) + 1 = 1$.
\end{proof}
\end{proposition}

\begin{lemma}
\label{vvprop99}
Define the parameters
$$
\Delta = \max_{s \in S_n} |W(s)|, \qquad \tilde \Delta = \max_{\substack{s \in \tilde S_n \setminus \{ \bot \}}}  |W(s)|.
$$
\begin{enumerate}
\item For any time $t$ and non-reject state $s \in \tilde S_t$, there holds $| T^{\tilde F}_{t,n}(s, \tilde W) - T^F_{t,n}(s, \tilde W)| \leq \delta \tilde \Delta$.

\item For any starting state $s$, there holds $| T^{\tilde F}_{0,n}(s, \tilde W) - T^F_{0,n}(s, W)| \leq \delta \Delta.$

\item For any state $s \in \tilde S_t \setminus \{ \bot \}$, there holds $$
{\mathfrak C}^{\tilde F}(s) \leq  \mathfrak C^F(s) + 5 \Delta \delta + \tilde \Delta \lambda_t.$$

\item For any starting state $s$, there holds $$
\mathfrak V^{\tilde F}( s ) \leq \mathfrak V^F(s) + 5 n \Delta \delta + \tilde \Delta \sum_{t=0}^{n-1} \lambda_t.
$$
\end{enumerate}
\end{lemma}
\begin{proof}
\begin{enumerate}
\item Let $z, \tilde z$ be the final states for automata $F, \tilde F$ under a common drivstream drawn from $\Omega_{t,n}$, so that $T^F_{t,n}(s,\tilde W) = \E[ \tilde W(z) ]$ and $T^{\tilde F}_{t,n}(s,\tilde W) = \E[ \tilde W( \tilde z) ]$.  

If $G(s) \geq 91$, then by \Cref{couple-prop} we have $\Pr( \phi(z) = \phi(\tilde z) = 0 ) \geq 1 - \delta$; when this holds we have $\tilde W(z) = \tilde W( \tilde z) = 0$. Similarly, if $G(s) < 91$, then $\Pr(z = \tilde z \neq \bot) \geq 1-\delta$, in which case $W(z) = \tilde W( \tilde z)$. We also have $|\tilde W (z) |, |\tilde W (\tilde z) | \leq \tilde \Delta$ always.

\item Again, let $z, \tilde z$ be the final states for automata $F, \tilde F$ from a starting state $s$ under a common drivestream drawn from $\Omega_{0,n}$. By (G1) we have $G(s) = 0$ so by \Cref{couple-prop} there is a probability of at least $1 -\delta$ that $z = \tilde z$ and $\phi(z) = \phi(\tilde z) = 1$, in which case  $\tilde W( \tilde z) = W( z)$. We also have $ |W(z)|, \tilde W( \tilde z)| \leq \Delta$.

\item Let $r_1, r_2 \in \Omega_t$; for $i = 1,2$ let $s_i = \tilde F(s,r_i)$ and let $z_i, \tilde z_i$ be the final states for automata $F, \tilde F$ reached from $s_i$ under a common drivestream drawn from $\Omega_{t+1, n}$. We need to show that
$$
\E \bigl[ |\tilde W( \tilde z_1) - \tilde W( \tilde z_2) |  \bigr] \leq \mathfrak C^F(s) + 5 \Delta \delta + \tilde \Delta \lambda_t 
$$

By Property (G2), since $z_1, z_2$ are obtained by stepping under drivestreams which differ only in coordinate $t$, we have $|\phi(z_1) - \phi(z_2)| \leq | G(z_1) - G(z_2) | \leq \lambda_t$.

If $G(s) \geq 92$, then by Property (G2) we have $G(s_1) \geq 91, G(s_2) \geq 91$. Then $\phi(\tilde z_1) = \phi(\tilde z_2) = 0$ with probability at least $1-2 \delta$. Since $|\tilde W(\tilde z_1) - \tilde W(\tilde z_2) | \leq 2 \tilde \Delta$ always, we overall have ${\mathfrak C}^{\tilde F}(s) \leq  4 \tilde \Delta \delta$, and the bound is clear.

So suppose $G(s) < 92$. Property (G2) gives $G(s_1) \leq 93, G(s_2) \leq 93$. So, with probability at least $1-2 \delta$ we have $\tilde z_1 = z_1 \neq \bot$ and $\tilde z_2 = z_2 \neq \bot$. Since again  $|\tilde W(\tilde z_1) - \tilde W(\tilde z_2) | \leq 2 \tilde \Delta$ always, we have shown so far that
$$
  \E \bigl[  | \tilde W( \tilde z_1) - \tilde W( \tilde z_2) | ] \leq  \E \bigl[ | \phi(z_1) W(  z_1) - \phi(z_2) W(  z_2) | ] +  4 \tilde \Delta \delta
$$

We can in turn factor this as 
\begin{align*}
&  \E[ | \phi(z_1) W(  z_1) - \phi(z_2) W(  z_2) | ]  = \E[ | \phi(z_2) ( W(  z_1) - W(  z_2)) + (\phi(z_1) - \phi(z_2)) W(z_1) | ] \\
&\qquad \qquad \leq  \E[ \phi(z_2) \cdot | W(  z_1) - W(  z_2)|  ] + \E[ | \phi(z_1) - \phi(z_2)| \cdot  |W(z_1) | ] \\
&\qquad \qquad \leq  \E[ | W(  z_1) - W(  z_2)|  ] + \lambda_t \E[ | W(z_1) | ]
 \end{align*}

The first term is at most $\mathfrak C^F(s)$, and $\E \bigl[ |W(z_1)| \bigr] \leq \tilde \Delta + \Delta \Pr( G(z_1) > 100) \leq \tilde \Delta + \delta \Delta$.

\item We use part (c) and calculate here:
\begin{align*}
\mathfrak V^{\tilde F}(s) &= \sum_{h=1}^{n} \max_{\vec r \in \Omega_{0,h}} \mathfrak C^{\tilde F}(   F(s, \vec r)) \leq  \sum_{h=1}^{n} \max_{\vec r \in \Omega_{0,h}} (\mathfrak C^F(F(s,\vec r)) + 5 \Delta \delta + \tilde \Delta \lambda_{h-1}) \\
&= 5 n \Delta \delta  + \tilde \Delta \sum_{t=0}^{n-1} \lambda_t +   \sum_{h=1}^{n} \max_{\vec r \in \Omega_{0,h}} \mathfrak C^F(F(s,\vec r)) = \mathfrak V^F(s) + 5 n \Delta \delta + \tilde \Delta \sum_{t=0}^{n-1} \lambda_t. \qedhere
\end{align*} 
\end{enumerate}
\end{proof}

\subsection{Guard constructions}
The simplest type of guard applies to a counter computing a running sum of the form 
$$
\sum_t f_t( R_t) \qquad \text{for functions $f_t:  \Omega_t \rightarrow \mathbb Z$}
$$

 They are ubiquitous in randomized algorithm.  We can analyze their behavior and construct guards via Bennet's concentration inequality for sums of independent random variables \cite{B62}. Let us define some relevant parameters:
$$
\mu_t = \E_{r \sim \Omega_t}[ f_t(r) ], \quad
M_t = \max_{r \in \Omega_t} |f_t(r) - \mu_t |, \quad \kappa  = \sum_{t=0}^{n-1} \Var_{r \sim \Omega_t}[f_t(r)], \quad M = \max_{t \in [n]} M_t
$$

It is straightforward to show the following bound:
\begin{lemma}
\label{counterxx}
For $B \geq  1000  (1 + M + \sqrt{\kappa}) \log(n/\delta)$, the function $G(s)  = |s - \mu_t|/B$ defines a guard for the counter, with concentration parameter $\delta$ and Lipschitz vector $\lambda_t = \frac{M_t}{B}$. The resulting automaton $\tilde F$ has statesize $\tilde \eta \leq O( B )$. 

In this context, we refer to $B$ as the \emph{span} of the automaton.
\end{lemma}

In a slightly more general setting, we may need to keep track of multiple counters, combined via a non-linear weight function (see, for example, the MAX-CUT algorithm described later). Generically, let us define the \emph{Cartesian product} of two automata $F^1, F^2$ to be a combined automaton $F = F^1 \times F^2$, with state space $S_t = S^1_t \times S^2_t$, and with transition rule $$
F( (s_1, s_2),r ) = (F^1( s_1, r), F^2(s_2,r))
$$

Critically, our use of guards allows us to transport truncated automata for $F_1, F_2$  into a truncated automaton for $F$ itself.

\begin{lemma}
\label{cartesian-lemma}
Suppose that $G_1, G_2$ are guards for $F_1, F_2$, with concentration bounds $\delta_1, \delta_2$ and Lipschitz vectors $\lambda^1, \lambda^2$ respectively. Then, for a state $s = (s_1, s_2)$ of automaton $F = F_1 \times F_2$, the function $G(s) = \max\{ G_1(s_1), G_2(s_2) \}$ defines a guard with concentration parameter $\delta_1 + \delta_2$ and Lipschitz vector $\lambda_t = \max \{ \lambda^1_t, \lambda^2_t \}$. 

Furthermore, if the truncated automata $\tilde F_1, \tilde F_2$ have statesizes $\tilde \eta_1$ and $\tilde \eta_2$ respectively, then the truncated automaton $\tilde F$ has statesize $\tilde \eta \leq \tilde \eta_1 \tilde \eta_2$.
\end{lemma}

\section{Applications}
We present the derandomization of two basic problems: the Gale-Berlekamp Switching Game (\cref{sec:bgame}) and approximate MAX-CUT via SDP rounding (\cref{sec:max-cut}).  Through our new construction, we improve the deterministic processor complexity for both problems.

\subsection{Gale-Berlekamp Switching Game}\label{sec:bgame}
The Gale-Berlekamp Switching Game is a classic combinatorial problem: given an $n \times n$ matrix $A$ with entries $A_{ij} = \pm 1$, we want to find vectors $x, y \in \{-1, +1 \}^n$ to maximize the imbalance $I = \sum_{i,j} A_{i,j} x_i y_j$.  Our new construction applies in a straightforward way to this problem.

\begin{theorem}\label{thm-gb_game}
 There is a parallel deterministic algorithm to find $\vec x, \vec y \in \{-1, +1\}^n$ with imbalance
    $I \geq (\sqrt{2/\pi} - o(1)) n^{3/2}$ using $\tilde O (n^{3.5} )$ processors.
\end{theorem}

In order to show \Cref{thm-gb_game}, following \cite{B97} and \cite{alon-spencer}, we set $x_i = 1$ if $\sum_j A_{ij} y_j > 0$, and $x_i = -1$ otherwise. This gives
    $$
    I =  \sum_{i,j} A_{i,j} x_i y_j = \sum_i x_i \sum_j A_{i,j} y_j = \sum_i \Bigl| \sum_j A_{i,j} y_j \Bigr|.
    $$
    
    As shown in \cite{BS71}, for $\vec y$ uniformly drawn from $\Omega = \{-1, +1 \}^n$, there holds $\E[|\sum_j A_{ij} y_j|] \geq \sqrt{2 n/\pi} - o(\sqrt{n})$ for each $i$. This is the statistical test we want to fool. Accordingly, we take $\Omega_t$ to be uniform distribution on $\{-1, 1 \}$ and set $y_t = R_t$.   For each $i$, we have a counter automaton $F_i$ to track the sum $\sum_t A_{it} R_t$, with weight function $W(s) = |s|$.     Since $\mathfrak C^{F_i}(s) \leq 2$ for each state $s$, the automaton $F_i$ has total variability $2 n$. 
    
    We construct a truncated automaton $\tilde F_i$ via \Cref{counterxx} with $\delta = n^{-10}$, where $M = 1$ and $\kappa = \sum_j A_{ij}^2 = n$ and $B = \Theta(\sqrt{n} \log n)$ and $\lambda_t = \tilde O(1/\sqrt{n})$. Here $\Delta \leq n$ and $\tilde \Delta \leq \tilde O( \sqrt{n})$. By \Cref{vvprop99}, the automaton $\tilde F_i$ has statesize $\tilde O(\sqrt{n})$, has  total variability
$$
\mathfrak V^{\tilde F_i}(s) \leq \mathfrak V^{F_i}(s) + 5 n \Delta \delta + \tilde \Delta \sum_{t=0}^{n-1} \lambda_t \leq \tilde O(n)
$$
and the final state $\tilde z$ satisfies
$$
\E[\tilde W_i(\tilde z)] \geq \E[W_i(z)] - \delta \Delta \geq \sqrt{2 n/\pi} - o(\sqrt{n}).
$$

The complexity now follows  from \cref{multi-aprop} with $\eps = 1/\sqrt{n \log n}$, where we have $n$ automata each of statesize $\hat \eta_i = \tilde O(\sqrt{n})$. For any row $i$, the final state $\tilde z$ satisfies
\begin{align*}
\E_D \Bigl[ |\sum_j A_{ij} y_j | \Bigr] \geq \E_D \bigl[ \tilde W_i(z) \bigr] \geq \E_{\Omega} \bigl[ W_i(\tilde z) \bigr] - o(\sqrt{n}) \geq  \sqrt{2 n/\pi} - o(n)
\end{align*}
By searching the space exhaustively, we can find a specific sequence $\vec y$ with the desired imbalance bounds.    This concludes the proof of \Cref{thm-gb_game}.

\subsection{Approximate MAX-CUT}\label{sec:max-cut}
Consider a graph with $n$ vertices and $m$ edges. The MAX-CUT problem is to find a vertex set maximizing the total weight of edges with exactly one endpoint in it. The seminal work of Goemans \& Williamson \cite{GW95} showed that MAX-CUT can be approximated to a factor of $\alpha \approx 0.878$ by rounding a semi-definite program (SDP). Moreover, the  integrality gap of the SDP is precisely $\alpha$ \cite{FS02}, and assuming the Unique Games Conjecture no better approximation ratio is possible for polynomial-time algorithms \cite{khot2007optimal}.

To review, note that MAX-CUT can be formulated as the following integer program:
\begin{equation*}
\max \quad \frac{1}{2} \sum_{(i,j) = e \in E} w_{e}(1 - v_i \cdot v_j) 
\quad \text{s.t.} \quad v_i \in \{-1, 1\}: i \in V.
\end{equation*}

This can be relaxed to the following SDP:
\begin{equation*}
\max \quad \frac{1}{2} \sum_{(i,j) = e \in E} w_{e}(1 - v_i \bullet v_j) 
\quad \text{s.t.} \quad v_i \in [-1,1]^n, ||v_i||_2 = 1: i \in V,
\end{equation*}
where $\bullet$ denotes the inner product in $\mathbb R^n$.

The SDP relaxation  itself can be approximated in polylogarithmic time and near-linear work by \cite{ALO16}, so we will be concerned with rounding the solution $v$.  For the randomized algorithm, this is done by drawing independent standard Gaussian random variables $X_0, \dots, X_{n-1}$ and producing the cut $\mathcal C = \{ i \mid v_i \bullet X \geq 0 \}$.  The resulting (random) cutsize $W$ satisfies
\begin{align*}
    \mathbb{E}[W] &= \sum_{(i,j) = e \in E} w_{e} \Pr[(i,j) \text{ is cut}] = \sum_{(i,j) = e \in E} w_{e} \Pr[\sgn(v_i \bullet X) \neq \sgn(v_j \bullet X)] \\  &
    = \sum_{(i,j) = e \in E} w_{e} \frac{\arccos(v_i \bullet v_j)}{\pi} \geq \alpha \cdot \text{OPT},
\end{align*}
where OPT denotes the size of the maximum cut. Following \cite{S02}, we will derandomize this by fooling a statistical test for each edge $ij$.

\begin{theorem}
\label{maxcut-thm}
     There is a deterministic parallel algorithm to find an $\alpha (1 - \eps)$-approximate rounding of a MAX-CUT SDP solution with $\tilde O( \frac{m n^3}{\eps^4})  $ processors.
\end{theorem}
Via the algorithm of \cite{ALO16} to solve the SDP, this gives the following main result:
\begin{corollary}
For arbitrary constant $\eps > 0$, we can find an $\alpha(1 - \eps)$-approximate MAX-CUT solution using $\tilde O( m n^3 )$ processors and polylogarithmic time.
\end{corollary}

The remainder of the section is devoted to showing \Cref{maxcut-thm}. To avoid cumbersome calculations, we will show an algorithm for a $(1 - O(\eps))$-approximation, assuming $\eps$ is smaller than any needed constants. For readability, we write $\tilde O()$ throughout to suppress any $\polylog(n/\eps)$ terms.

\bigskip

To apply our derandomization framework, we must first quantize the continuous Gaussian variables. Concretely, let us fix the quantization parameter $$
\gamma = \frac{\eps}{C \sqrt{n \log(n/\eps)}} \qquad \text{for a sufficiently large constant $C > 0$.}
$$

Each $R_t$ is derived by truncating a standard Gaussian to within $\pm C \log(n/\eps)$ and then rounding to the nearest multiple of $\gamma$. Then, each term  $v_{i,t} R_t$ is rounded again to the nearest multiple of $\gamma$.  Since this comes up frequently in the analysis, we denote the ``rounded" inner product by:
$$
v \star r = \gamma \sum_t \big \lceil \tfrac{1}{\gamma} v_{t} r_t \big \rfloor
$$

We will form our cut set to be $\mathcal C = \{ i: v_i \star R \geq 0 \}$.

\begin{lemma}
\label{coupling-prop}
There is a coupling between random variables $X$ and $r$ such that, for any $u \in \mathbb R^n$ with $||u||_{\infty} \leq 1$, there holds $\Pr \bigl( | u \bullet X -  u \star R | > \eps  \bigr) \leq (\eps/n)^{10}.$
\end{lemma}
\begin{proof}
If $|X_t| \leq C \log(n/\eps)$, set $R_t$ to be the multiple of $\gamma$ closest to $X_t$; otherwise, set $R_t = 0$.

 Let $Y_t = u_t X_t - \gamma \lceil \tfrac{1}{\gamma} u_t R_t \rfloor$, so that $u \bullet X -  u \star R = \sum_t Y_t$. Note that when $|X_t| \leq C \log(n/\eps)$ we have $|Y_t| \leq 2 \gamma$.
Furthermore, since we round to the nearest value, the distribution of $Y_t$ is symmetric around zero. So $\sum_t Y_t$ is a sum of $n$ independent random variables, each of mean zero and bounded in the range $\pm 2 \gamma$ (barring the extraordinary event that $|X_t| > C \log(n/\eps)$). By Hoeffding's inequality, and taking into account the possibility of a larger deviation, we thus have
$$
\Pr \bigl( | u \bullet X -  u \star R | > \eps \bigr) \leq n \cdot e^{-\Omega(C \log(n/\eps))} + e^{ -\Omega \bigl( \frac{\eps^2 }{(2 \gamma)^2 n} \bigr)}
$$

For $\gamma = \frac{\eps}{C \sqrt{n \log(n/\eps)}}$ and large enough $C$,  this is at most $(\eps/n)^{10}$.
\end{proof}

We construct an automaton $F_{ij}$ for each edge $ij \in E$ to  keep track of the running sums $c_i =  v_i \star R$ and $c_j = v_j \star R$. It uses the following weight function for the final state $z = (c_i, c_j)$:
$$
W_{ij}(z) = \begin{cases} 
0 & \text{if $\sgn(c_i) = \sgn(c_j)$} \\
\min \{1, |c_i|/\eps, |c_j|/\eps \} &  \text{if $\sgn(c_i) \neq \sgn(c_j)$}
\end{cases}
$$

Note that the edge $ij$ is cut if and only if $W_{ij}(z) > 0$. We remark that the obvious weight function would simply be the indicator function that $\sgn(c_i) \neq \sgn(c_j)$. Our weight function serves as a pessimistic estimator, with much better Lipschitz properties.

\begin{proposition}
\label{ija-prop0}
For $\vec R \sim \Omega$, the final state $ z = F_{ij}(\vec R, 0)$  satisfies $\E[W (z) ] \geq \frac{\arccos ( v_i \bullet v_j) }{\pi}  - O(\eps)$. 
\end{proposition}
\begin{proof}
Let $\mathcal E$ denote the event that $\sgn(v_i \star R) \neq \sgn(v_j \star R)$ and $|v_i \star R| > \eps$ and $|v_j \star R| > \eps$. When $\mathcal E$ occurs, we have $W(z) = 1$. So $\E[W(z)] \geq \Pr(\mathcal E)$.   We use the coupling in \Cref{coupling-prop} to calculate
        \begin{align*}
  &      \Pr( \mathcal E ) = 2 \Pr( v_i \star R > \eps \wedge v_j \star R < -\eps ) \tag {by symmetry} \\
        & \quad \geq 2 \Pr( v_i \bullet X > 2 \eps \wedge v_j \bullet X < -2 \eps ) - 2 (\eps/n)^{10}  \tag{by \Cref{coupling-prop}} \\
        & \quad \geq 2 \Pr( v_i \bullet X > 0 \wedge v_j \bullet X < 0) - \Pr ( |v_i \bullet X| < 2 \eps ) - \Pr (  |v_j \bullet X| < 2 \eps)  - O(\eps) \\
        & \quad \geq \Pr( \sgn(v_i \bullet X) \neq \sgn(v_j \bullet X) ) - O(\eps) 
        \end{align*}
        where the last inequality holds since $v_i \bullet X$ and $v_j \bullet X$ are standard Gaussian variables. As in the usual MAX-CUT argument, we have $\Pr( \sgn(v_i \bullet X) \neq \sgn(v_j \bullet X) )  = \frac{\arccos(v_i \bullet v_j)}{\pi}$.
\end{proof}

\begin{proposition}
\label{yuu1}
For any state $s$ at time $t$, the edge-automaton $F_{ij}$ satisfies
$$
\mathfrak C^{F_{ij}}(s) \leq \tilde O \Biggl(  \sqrt{ \frac{ v_{i,t}^2 }{\eps^2 + v_{i,t}^2 + \dots + v_{i, n}^2}} + \sqrt{ \frac{ v_{j,t}^2 }{\eps^2 + v_{j,t}^2 + \dots + v_{j, n}^2}}+ \eps/n \Biggr)
$$
\end{proposition}
\begin{proof}
 Let $A_{i} = \gamma \sum_{k < t} \lceil \frac{v_{i,k} R_{k}}{\gamma} \rfloor$ and $B_i =  \gamma  \sum_{k > t} \lceil \frac{ v_{i,k} R_{k}}{\gamma} \rfloor$, so that $c_i = A_i + B_i + \gamma \lceil \frac{v_{i,t} R_t}{\gamma} \rfloor$; similarly define $A_j$ and $B_j$.  When  computing $\mathfrak C^{F_{ij}}(s)$, we look at the expected change in the final state's weight, assuming all variables up to time $t$ are fixed while variables $R_{t+1}, \dots, R_{n-1}$ are drawn independently from $\Omega$. In particular, $A_i, A_j$ are fixed while $B_i, B_j$ retain their original distribution.  

Let $\nu = 2 (|v_{i,t}| + |v_{j,t}|) \cdot C \log(n/\eps)$. Note that changing $R_t$ can change the final weight $W(z)$ by at most $\min\{1,  \nu/\eps\}$. Furthermore, if $|A_i + B_i| > \eps + \nu$ and $|A_j + B_j| > \eps + \nu$, then $|c_i|, |c_j| > \eps$ irrespective of $R_t$. In such a case, changing $R_t$ has no effect on $W(z)$. So
\begin{align*}
\mathfrak C(s) &\leq \min\{1, \nu/\eps \} \cdot \Pr \bigl( |A_i + B_i| \leq \eps + \nu \vee |A_j + B_j| \leq \eps + \nu \bigr) 
\end{align*}
where the probability is taken over drawing $R_{t+1}, \dots, R_n \sim \Omega$.

Define vector $v'_i$ by zeroing out coordinates $0, \dots, t$ of $v_i$. So $B_i = v'_i \star R$, and $|A_i+B_i| \leq \eps + \nu$ if and only if $v'_i \star R \in [ -A_i -  \eps - \nu,  -A_i +  \eps + \nu]$. From the coupling of \Cref{couple-prop}, we get \begin{align*}
&\Pr \bigl( |A_i + B_i| \leq \eps + \nu \bigr) =  \Pr \bigl( v'_i \star R \in [ -A_i -  \eps - \nu,  -A_i +  \eps + \nu] \bigr) \\
& \qquad  \leq \Pr \bigl( v'_i \bullet X\in [ -A_i -  2 \eps - \nu,  -A_i +  2 \eps + \nu ] \bigr) + (\eps/n)^{10}
\end{align*}

Since $v'_i \bullet X$ is a Gaussian with mean zero and variance $\sigma^2 = v_{i,t+1}^2 + \dots, v_{i,n-1}^2$, the probability it lies in a region of width $4 \eps + 2 \nu$ is at most $\min\{1, \frac{ 4 \eps + 2 \nu}{ \sqrt{2 \pi \sigma^2}} \}$. Overall, we see that
$$
\Pr_{R_{t+1}, \dots, R_n \sim \Omega} \bigl( |A_i + B_i| \leq \eps + \nu \bigr) \leq  O\Bigl( \min\{1, \frac{\eps + \nu}{ \sqrt{v_{i,t+1}^2 + \dots, v_{i,n-1}^2}} \} \Bigr) + (\eps/n)^{10}.
$$

A completely analogous bound holds for $v_j$. Overall, we have
\begin{align*}
\mathfrak C(s) &\leq \min\{1, \nu/\eps\} \cdot O\Bigl( \min\{1, \frac{\eps + \nu}{ \sqrt{v_{i,t+1}^2 + \dots, v_{i,n-1}^2}} \} +  \min\{1, \frac{\eps + \nu}{ \sqrt{v_{j,t+1}^2 + \dots, v_{j,n-1}^2}} \} + (\eps/n)^{10} \Bigr)
\end{align*}

After some rearrangement, this gives the claimed bound.
\end{proof}

\begin{lemma}
\label{yuu0}
The edge automaton $F_{ij}$ has total variability $\tilde O( \sqrt{n} )$.
\end{lemma}
\begin{proof}
The starting state $0$ satisfies 
$$
\mathfrak V^{F_{ij}}(0) \leq \sum_{t=0}^{n-1} \max_{s \in S_t} \mathfrak C^{F_{ij}}(s) \leq \tilde O \Biggl( \sum_{t=0}^{n-1}  \sqrt{ \frac{ v_{i,t}^2 }{\eps^2 + v_{i,t}^2 + \dots + v_{i, n-1}^2}} + \sqrt{\frac{ v_{j,t}^2 }{\eps^2 + v_{j,t}^2 + \dots + v_{j, n-1}^2}}+ \eps/n \Biggr).
$$

We show in \Cref{combb1} in the Appendix (taking $u = \eps^2, p = 1/2$, and $a_k = v_{i,n-k}^2: k = 1, \dots, n$) that
$$
\sum_{t=0}^{n-1} \sqrt{ \frac{ v_{i,t}^2 }{\eps^2 + v_{i,t}^2 + \dots + v_{i, n-1}^2}} \leq O( \sqrt{n \log(1/\eps)} )
$$

A completely analogous bound holds for the vector $v_j$.
\end{proof}

There are two further optimizations to reduce the processor complexity.
First, we apply the method of \Cref{sec:counter} to reduce the automata state space. Second, we take advantage of the counter structure to efficiently compute the transition matrices.

\begin{proposition}
\label{automaton-fi-prop}
For the automaton $F_{ij}$, there is a guard $G_{ij}$ with parameters $\delta = (\eps/n)^{10}$ and  $\lambda_t = \tilde O(v_{i,t} + v_{j,t})$. Its truncated automaton $\tilde F_{ij}$ has statesize $\tilde O(n/\eps^2)$.
\end{proposition}
\begin{proof}
For each vertex $i$, we can construct an automaton $F_i$ to track $c_i = v_i \star R$. The scaled sum $\frac{1}{\gamma} (v_i \star R)$ is an integer-valued counter, so we can apply the construction of \Cref{counterxx}. Since $|R_t| \leq \tilde O(1)$ and $v_i$ is a unit vector, we get $M_t \leq \tilde O( |v_{i,t}|/\gamma )$ and $M \leq \tilde O(1/\gamma)$.  We can also calculate $$
\kappa = \tfrac{1}{\gamma^2} \E[ (v_i \star R)^2 ] \leq O(1/\gamma^2) \sum_t \E[ (v_{i,t} R_t)^2 ] \leq \tilde O(1/\gamma^2) \sum_t  v_{i,t}^2 = \tilde O(1/\gamma^2).
$$

With these parameters, we can choose $\delta' = \delta/2$ and span $B = \ceil{1000 (M + \sqrt{\kappa}) \log(n/\delta')}  = \tilde \Theta(1/\gamma)$.  This gives a statesize of $\tilde O(1/\gamma)$ for automaton $\tilde F_i$. 

Note that $F_{ij}$ is simply the Cartesian product $F_i \times F_j$. Given the automata $F_i, F_j$ with their guards, we use \Cref{cartesian-lemma} to build a guard for automaton $F_{ij}$; it has concentration $2 \delta' = \delta$, has Lipschitz vector $\lambda_t = \tilde O( v_{i,t} + v_{j,t} )$, and has statesize $\tilde O(1/\gamma^2) = \tilde O(n/\eps^2)$.
\end{proof}

\begin{proposition}\label{prop-maxcut_omega}
For the automaton $F_{ij}$, we can compute all transition matrices $T_{t,t+h}$ for $h$ a power of two and $t$ divisible by $h$ using $\tilde O( n^3 / \eps^2)$ processors.
\end{proposition}
\begin{proof}
Since $F_{ij}$ is a pair of counters, $T_{t,t+h}(s,s')$ only depends on the \emph{difference} $s' - s$; the starting state $s$ is irrelevant. Instead of storing the full matrix $T_{t,t+h}$, we only need to keep track of the probability distribution  on pairs
\begin{equation}
\label{tragg2}
\bigl( \sum_{k = t}^{t+h - 1} \lceil  v_{i,k} R_k / \gamma \rfloor, 
\sum_{k = t}^{t+h - 1} \lceil v_{j,k} R_k / \gamma \rfloor \bigr)
\end{equation}

In particular, the transition matrix $T_{t,t+2h}$ is given as a two-dimensional convolution of the distributions for $T_{t,t+h}$ and $T_{t+h,t+2h}$.  It is well-known that convolutions can be computed in near-linear work and logarithmic depth via Fast Fourier Transform (see e.g., \cite{L14}). Thus, each distribution $T_{t,t+2h}$ can be obtained from $T_{t,t+h}, T_{t+h, t+2h}$ with a processor count of $\tilde O(n^2/\eps^2)$. The total processor count is at most $\tilde O(n^3/\eps^2)$.
\end{proof}

\begin{proposition}
\label{geew}
Automaton $\tilde F_{ij}$ has total variability $\tilde O(\sqrt{n})$. 

Its final state $\tilde z$ has $\E[ \tilde W_{ij} (\tilde z) ] \geq  \frac{\arccos(v_i \bullet v_j)}{\pi} - O(\eps)$.

For any state $s$, we have  $| T^{\tilde F_{ij}}(s, \tilde W) - T^{F_{ij}}(s, \tilde W)| \leq \delta$.
\end{proposition}
\begin{proof}
We have already shown that $F_{ij}$ has total variability $\tilde O(\sqrt{n})$ and its final state $z$ has  $\E[W(z)] \geq \frac{\arccos(v_i \bullet v_j)}{\pi} - O(\eps)$. We apply  \Cref{vvprop99} to transfer these bounds to $\tilde F_{ij}$, noting that $\tilde \Delta \leq \Delta \leq 1$ and $\sum_t \lambda_t = \tilde O( \sum_t |v_{i,t}| + |v_{j,t}|) \leq \tilde O(\sqrt{n})$.
\end{proof}

\paragraph{Finishing up} We are now ready to compute the total algorithm complexity. 

 Each automaton $\tilde F_{ij}$ has total variability $\tilde O(\sqrt{n})$ and statesize $\tilde O(n/\eps^2)$.  By \Cref{prop-generic_omega} and \Cref{prop-maxcut_omega}, we compute approximate transition vectors $$
\hat V_t(s) = T_{t,n}^{F_{ij}}(s,\tilde W_{ij})
$$
in a total of $\tilde O(n^3/\eps^2)$ processors.  Although these are computed for $F_{ij}$, they achieve error $\beta = \delta$ when compared to the transition matrices for the truncated automaton $\tilde F_{ij}$.

Overall we have $\sigma = \tilde O(\sqrt{n}/\eps)$ and $\eta = \tilde O( m n / \eps^2)$. We run \textsc{Fool} with parameter  $\eps' = \tilde \Theta(\eps/\sqrt{n}) $,  with total cost $\tilde O (m  n^3/\eps^4  + m n^2/\eps^3)$.

Now each edge $ij$ is only cut if the final state $z$ satisfies $\tilde W_{ij}(z) > 0$. By \Cref{prop-add-err-add} and \Cref{geew}, the resulting distribution $D$ satisfies
$$
\Pr_{\vec R \sim D}( \text{edge $ij$ cut} ) \geq \E_{\vec R \sim D} [ \tilde W_{ij}(z) ] \geq \E_{\vec R \sim \Omega} [ \tilde W_{ij}(z) ] - O(\eps) \geq \frac{\arccos( v_i \bullet v_j)}{\pi} - O(\eps)
$$

Summing over all edges, the  expected weight of the cut edges is 
$$
\sum_{ij \in E} w_e \frac{\arccos( v_i \bullet v_j)}{\pi} - O(\eps) \sum_{ij \in E} w_e
$$

The first term is at least $\alpha \cdot \text{OPT}$. The second term is at most $O(\eps \cdot \text{OPT})$ since $\text{OPT} \geq \sum_e w_e/2$.

By searching the distribution $D$ exhaustively, we can find a cut satisfying these bounds.

\section{Acknowledgements}
Thanks to Mohsen Ghaffari and Christoph Grunau for explaining the works \cite{GG23} and \cite{GGR23}.
\appendix

\section{Some omitted proofs}
\label{app44}

\begin{proposition}
\label{app440}
For any time $t$ and state $s \in S_t$ and input $r \in \Omega_t$ we have
$$
|V_{t+1}(F(s,r)) - V_t(s)| \leq \mathfrak C(s)
$$
\end{proposition}
\begin{proof}
Expanding the definition of $V_{t}$, we have 
$$
V_{t}(s)  = T_{t,n}(s,W) = \sum_{r'} p_{\Omega_{t}}(r') T_{t+1, n}(F(s',r'), W) = \sum_{r'} p_{\Omega_{t}}(r') V_{t+1}(F(s,r'))
$$

 So $|V_{t+1}( F(s,r )) - V_{t}(t)|  \leq \sum_{r'} p_{\Omega_{t}}(r') |V_{t+1}( F(s,r )) - V_{t+1}( F(s,r'))|$. For such $r, r'$, let $z, z'$ be the final states obtained from  $F(s,r), F(s,r')$ respectively under a common drivestream drawn from $\Omega_{t+1, n}$.
Then $V_{t+1}(F(s,r)) = \E[ W(z) ]$ and $V_{t+1}(F(s,r')) = \E[W(z')]$. So 
$$
|V_{t+1}( F(s,r )) - V_{t+1}( F(s,r'))| = | \E[ W(z) - W(z') ] | \leq \E[ |W(z) - W(z')| ]
$$
which is at most $\mathfrak C(s)$ by definition.
\end{proof}

\begin{proof}[Proof of \Cref{prop-add-err0s}]
Since $V_{t+h}$ and $\hat V_{t+h}$ have entrywise difference at most $\beta$, we clearly have 
$$
\mathfrak S_{t,t+h} (s, \hat V_{t+h}) \leq \mathfrak S_{t,t+h}(s, V_{t+h}) + 2 \beta
$$

Now let $\vec r$ be a drivestream in $\Omega_{t,t+h}$, consider the state sequence $s_t = s, s_{t+1} = F(s_t, r_t), s_{t+2} = F(s_{t+1}, r_{t+1}), \dots, s_{t+h} = F(s,\vec r)$. By \Cref{app440}, we have 
\begin{align*}
|  V_{t+h}(F(s,\vec r)) -V_t(s)|  &\leq \sum_{k=t}^{t+h-1} | V_{k+1}(s_{k+1}) - V_{k}(s_{k}) | \leq \sum_{k=t}^{t+h-1} \mathfrak C(s_k).
\end{align*}

Since any such state $s_k$ is reachable from $s$, and since this holds for any drivestream $\vec r$, this implies that
$$
\mathfrak S_{t,t+h} (s, V_{t+h}) \leq 2 \sum_{k=t+1}^{t+h} \max_{\vec r \in \Omega_{t,k}} \mathfrak C( F(s,\vec r)).
$$

The claim follows.
\end{proof}

For the MAX-CUT application, we have the following useful general lemma:
\begin{lemma}
\label{combb1}
Let $a_1, \dots, a_n, u$ be non-negative values and let $p \in (0,1]$. Then 
$$
\sum_{i=1}^n \Bigl( \frac{a_i}{u + a_1 + \dots  + a_i} \Bigr)^p  \leq n^{1-p} \log^{p} \Bigl( 1+\frac{a_1 + \dots + a_n}{u} \Bigr)
$$
\end{lemma}
\begin{proof}
Let $s = a_1 + \dots + a_n$. We first show the bound for $p = 1$. Here, we have
\begin{align*}
\sum_{i=1}^n  \frac{a_i}{u + a_1 + \dots  + a_i}  &\leq \sum_{i=1}^n  \int_{t=u + a_1 + \dots + a_{i-1}}^{t=u + a_1 + \dots + a_i } \frac{\mathrm{d}t}{t} = \int_{t=u}^{t=a_1 + \dots + a_n + u} \frac{\mathrm{d}t}{t} \\
&= \log \frac{ a_1 + \dots + a_n + u}{u} = \log (1 + s/u)
\end{align*}

Next, for $p < 1$, we use Jensen's inequality applied to the concave-down function $x \mapsto x^p$:
\begin{align*}
\sum_{i=1}^n  \Bigl( \frac{a_i}{u + a_1 + \dots  + a_i} \Bigr)^p &\leq n \cdot \Bigl( \frac{1}{n} \sum_{i=1}^n \frac{a_i}{u + a_1 + \dots  + a_i} \Bigr)^p \\
&\leq n \cdot \Bigl( \frac{1}{n}  \log(1 + s/u) \Bigr)^p = n^{1-p} \log^p(1 + s/u). \qedhere
\end{align*}
\end{proof}

\bibliographystyle{alpha}
\bibliography{references}

\end{document}